\documentclass{article}
\usepackage[utf8]{inputenc}
\usepackage{amssymb}
\usepackage{amsmath}
\usepackage{bm}
\usepackage{amsthm}
\usepackage{bbm} 
\usepackage{color}
\usepackage{nicefrac}
\usepackage{graphicx}

\usepackage{tikz}

\title{A case study in almost-perfect security for unconditionally secure communication}
\author{Esteban Landerreche\footnote{\tt estebanlan@gmail.com}\\\\
David Fern\'andez-Duque\footnote{\tt david.fernandez@itam.mx}\\\\
Instituto Tecnol\'ogico Aut\'onomo de M\'exico\\
R\'io Hondo 1, 01080 Mexico City, Mexico
}

\begin{document}

\newtheorem{teorema}{Theorem}
\newtheorem{lema}{Lemma}
\newtheorem{corolario}{Corollary}
\newtheorem{define}{Definition}
\newtheorem{ejem}{Example}

\newcommand\ParEqui{m}
\newcommand\ParCath{\gamma}
\newcommand\CartaAlicia{\spadesuit}
\newcommand\CartaBob{\diamondsuit}
\newcommand\CartaCath{\clubsuit}
\newcommand\CartaDesc{\circ}
\newcommand\kapa[3]{[\dot{#3}]_{#1,#2}}
\newcommand\kkapa[3]{[\ddot{#3}]_{#1,#2}}
\newcommand\david[1]{{\color{red}DAVID: #1}}
\newcommand\esteban[1]{{\color{blue}ESTEBAN: #1}}




\maketitle
\begin{abstract}
In the Russian cards problem, Alice, Bob and Cath draw $a$, $b$ and $c$ cards, respectively, from a publicly known deck. Alice and Bob must then communicate their cards to each other without Cath learning who holds a single card. Solutions in the literature provide {\em weak security,} where Alice and Bob's exchanges do not allow Cath to know with certainty who holds each card that is not hers, or {\em perfect security,} where Cath learns no probabilistic information about who holds any given card. We propose an intermediate notion, which we call {$\varepsilon$-strong security,} where the probabilities perceived by Cath may only change by a factor of $\varepsilon$. We then show that a mild variant of the so-called {\em geometric strategy} gives $\varepsilon$-strong safety for arbitrarily small $\varepsilon$ and appropriately chosen values of $a,b,c$.
\end{abstract}

\section{Introduction}

Consider the following problem, which appeared in the 2000 Moscow Mathematics Olympiad:
\begin{quote}
Alice, Bob and Cath draw three, three and one cards, respectively, from a publicly known deck of seven. Alice and Bob wish to inform each other of the cards they hold, but they may only do so by public, unencrypted announcements. Moreover, they do not wish for Cath to know who holds a single card that is not hers. Can Alice and Bob achieve this?
\end{quote}
It later came to be known as the {\em Russian cards problem} \cite{hvd.studlog:2003}, and is interesting from a cryptographical viewpoint since it provides a framework for {\em unconditionally secure communication,} perhaps the strongest notion of security that one may demand from a cryptographic protocol.

\subsection{Notions of cryptographic security}

Claude Shannon, one of the first people to formalize the study of cryptography, proposed several notions of {\em cryptographic security.} To be precise, he defined the following:

\begin{itemize}
\item \textbf{Computational Security}: We say that a protocol is \textbf{computationally secure} for $n$ if at least $n$ operations are needed to break it. It is usually very difficult to prove that protocol is secure in this sense, as we would need to know all the possible strategies for attack. However, it is a good measure of when a system {\em isn't} secure, that is, when it fails to be secure for a relatively small $n$.

\item \textbf{Provable Security}: We say that a protocol is \textbf{provably secure} if we can link it with a `hard' problem, cryptographic or not, in such a way that solving the second problem will allow us to break the encryption. In that case, we know that we need at least as many operations to break the code as we need to solve the second problem. Typically, the `hard' problem is in {\sc np} but believed to not be in {\sc p}. Many of the cryptographic protocols in use today are based on this notion of security.

\item \textbf{Unconditional Security}: A protocol is \textbf{unconditionally secure} if it can't be broken even with unlimited computational resources; the eavesdropper simply does not have enough information to reconstruct the original message.
\end{itemize}
It should be clear that unconditional security implies both computational and provable security, and as such it would ideally be desirable to develop unconditionally secure cryptographic protocols. However, such protocols tend to be unpractical and as such few of them are known, with a notable example being Vernam's one time pad \cite{stinson2005cryptography}. However, the setup of the Russian cards, which presupposes a secure dealing phase, provides a convenient setup for developing unconditionally secure protocols.

\subsection{Related work}

The Russian cards problem may be traced back to Kirkman \cite{kirkman1847problem}, but recently it has received renewed attention after its inclusion in the 2000 Mathematics Olympiad \cite{hvd.studlog:2003}. One of the solutions for deals of distribution type $(3,3,1)$ uses the Fano plane, a special case of a combinatorial design, which can also be used for many other distribution types \cite{albert2003safe}. Another solution uses modular arithmetic, which can also be generalized for many distribution types where the eavesdropper holds one card \cite{cordonetal:2012}. These solutions use only two announcements, but some cases are known to require more. A solution using three announcements for $(4,4,2)$ is reported in \cite{threesteps}, and a four-step protocol for $c=O(a^2)$ and $b=O(c^2)$ is presented in \cite{cordon2013colouring}. The solution we will work with in this paper is similar to the one reported in \cite{cordon2013geometric}, which also takes two steps. The Russian cards problem has also been generalized to a larger number of agents in \cite{duan2010,sadi}.

However, while the protocols mentioned above provide unconditionally secure solutions to the Russian cards problem in that the eavesdropper may not {\em know with certainty} who holds a given card, that does not mean that she may not have a high probability of {\em guessing} this information correctly. To this end, stronger notions of security are studied in \cite{swanson2014combinatorial}. There, a distinction is made between {\em weak} and {\em perfect} security; in perfectly secure solutions, Cath does not acquire any probabilistic information about the ownership of any specific card. All of the above solutions provide weak security in this sense, but Swanson and Stinson show how designs may be used to achieve perfect security, an idea further developed in \cite{swanson2014additional}.

The solutions we present here will provide an intermediate level of security between weak and perfect, controlling the amount of probabilistic information that may be acquired by the eavesdropper, while having the advantage of being much easier to construct than perfectly secure solutions.

\section{A Worked Example}

We will motivate the work in this article with a relatively small example. Let's suppose we have $49$ cards, with Alice holding $7$, Cath holding $5$ and Bob the rest. In this case, Alice can take advantage of the fact that there is a field $\mathbb F_7$ with $7$ elements (the quotient $\mathbb Z/(7)$ forms a field), and thus may identify each point in the two-dimensional vector space over $\mathbb F_7$, which we will denote $\mathbbm{F}_7^2$, with a card. Moreover, she can do this in such a way that her cards (marked by $\CartaAlicia$) form a line. Suppose then that Cath holds the cards marked by $\CartaCath$, while Bob holds the rest of the cards ($\CartaBob$).

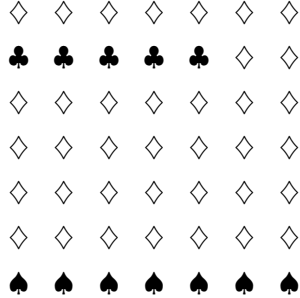
\begin{figure}[h]
\begin{center}
\begin{tikzpicture}[scale=0.6]
\node (0) at (0,0) {$\CartaAlicia$};
\node (1) at (1,0) {$\CartaAlicia$};
\node (2) at (2,0) {$\CartaAlicia$};
\node (3) at (3,0) {$\CartaAlicia$};
\node (4) at (4,0) {$\CartaAlicia$};
\node (5) at (5,0) {$\CartaAlicia$};
\node (6) at (6,0) {$\CartaAlicia$};

\node (01) at (0,1) {$\CartaBob$};
\node (11) at (1,1) {$\CartaBob$};
\node (21) at (2,1) {$\CartaBob$};
\node (31) at (3,1) {$\CartaBob$};
\node (41) at (4,1) {$\CartaBob$};
\node (51) at (5,1) {$\CartaBob$};
\node (61) at (6,1) {$\CartaBob$};

\node (02) at (0,2) {$\CartaBob$};
\node (12) at (1,2) {$\CartaBob$};
\node (22) at (2,2) {$\CartaBob$};
\node (32) at (3,2) {$\CartaBob$};
\node (42) at (4,2) {$\CartaBob$};
\node (52) at (5,2) {$\CartaBob$};
\node (62) at (6,2) {$\CartaBob$};

\node (03) at (0,3) {$\CartaBob$};
\node (13) at (1,3) {$\CartaBob$};
\node (23) at (2,3) {$\CartaBob$};
\node (33) at (3,3) {$\CartaBob$};
\node (43) at (4,3) {$\CartaBob$};
\node (53) at (5,3) {$\CartaBob$};
\node (63) at (6,3) {$\CartaBob$};

\node (04) at (0,4) {$\CartaBob$};
\node (14) at (1,4) {$\CartaBob$};
\node (24) at (2,4) {$\CartaBob$};
\node (34) at (3,4) {$\CartaBob$};
\node (44) at (4,4) {$\CartaBob$};
\node (54) at (5,4) {$\CartaBob$};
\node (64) at (6,4) {$\CartaBob$};

\node (05) at (0,5) {$\CartaCath$};
\node (15) at (1,5) {$\CartaCath$};
\node (25) at (2,5) {$\CartaCath$};
\node (35) at (3,5) {$\CartaCath$};
\node (45) at (4,5) {$\CartaCath$};
\node (55) at (5,5) {$\CartaBob$};
\node (65) at (6,5) {$\CartaBob$};

\node (06) at (0,6) {$\CartaBob$};
\node (16) at (1,6) {$\CartaBob$};
\node (26) at (2,6) {$\CartaBob$};
\node (36) at (3,6) {$\CartaBob$};
\node (46) at (4,6) {$\CartaBob$};
\node (56) at (5,6) {$\CartaBob$};
\node (66) at (6,6) {$\CartaBob$};

\end{tikzpicture}
\end{center}
\caption{Alice assigns each card to a point on the plane in such a way that her hand forms a line. She does not know how the other cards will fall, since she can only see her own hand. In this example, all of Cathy's cards happen to fall on another line.}
\label{FirstAnuncio}
\end{figure}

Alice then announces how she has distributed the cards on the plane. In this particular announcement, Cath's cards all fall within the same line. This is an extreme case, but it is a real possibility, as Alice has no knowledge of Cath's hand when she makes her announcement. Bob and Cath know that Alice's hand falls on a line, but they do not know which line. Bob then knows exactly which cards Alice holds (since there is only one complete line that he does not hold), but Cath does not. However, she may consider it more likely that Alice holds one card over another. To illustrate this, let us consider the points labeled $x$ and $y$ in Figure \ref{XYlines}.

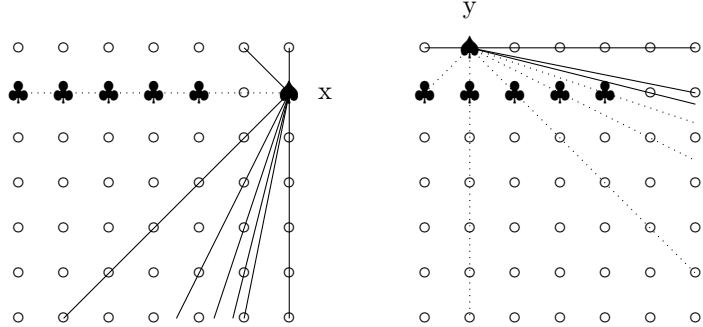
\begin{figure}
\begin{center}
\begin{tikzpicture}[scale=0.6]
\node (0) at (0,0) {$\CartaDesc$};
\node (1) at (1,0) {$\CartaDesc$};
\node (2) at (2,0) {$\CartaDesc$};
\node (3) at (3,0) {$\CartaDesc$};
\node (4) at (4,0) {$\CartaDesc$};
\node (5) at (5,0) {$\CartaDesc$};
\node (6) at (6,0) {$\CartaDesc$};

\node (01) at (0,1) {$\CartaDesc$};
\node (11) at (1,1) {$\CartaDesc$};
\node (21) at (2,1) {$\CartaDesc$};
\node (31) at (3,1) {$\CartaDesc$};
\node (41) at (4,1) {$\CartaDesc$};
\node (51) at (5,1) {$\CartaDesc$};
\node (61) at (6,1) {$\CartaDesc$};

\node (02) at (0,2) {$\CartaDesc$};
\node (12) at (1,2) {$\CartaDesc$};
\node (22) at (2,2) {$\CartaDesc$};
\node (32) at (3,2) {$\CartaDesc$};
\node (42) at (4,2) {$\CartaDesc$};
\node (52) at (5,2) {$\CartaDesc$};
\node (62) at (6,2) {$\CartaDesc$};

\node (03) at (0,3) {$\CartaDesc$};
\node (13) at (1,3) {$\CartaDesc$};
\node (23) at (2,3) {$\CartaDesc$};
\node (33) at (3,3) {$\CartaDesc$};
\node (43) at (4,3) {$\CartaDesc$};
\node (53) at (5,3) {$\CartaDesc$};
\node (63) at (6,3) {$\CartaDesc$};

\node (04) at (0,4) {$\CartaDesc$};
\node (14) at (1,4) {$\CartaDesc$};
\node (24) at (2,4) {$\CartaDesc$};
\node (34) at (3,4) {$\CartaDesc$};
\node (44) at (4,4) {$\CartaDesc$};
\node (54) at (5,4) {$\CartaDesc$};
\node (64) at (6,4) {$\CartaDesc$};

\node (05) at (0,5) {$\CartaCath$};
\node (15) at (1,5) {$\CartaCath$};
\node (25) at (2,5) {$\CartaCath$};
\node (35) at (3,5) {$\CartaCath$};
\node (45) at (4,5) {$\CartaCath$};
\node (55) at (5,5) {$\CartaDesc$};
\node[label=right:x] (65) at (6,5) {$\CartaAlicia$};

\node (06) at (0,6) {$\CartaDesc$};
\node (16) at (1,6) {$\CartaDesc$};
\node (26) at (2,6) {$\CartaDesc$};
\node (36) at (3,6) {$\CartaDesc$};
\node (46) at (4,6) {$\CartaDesc$};
\node (56) at (5,6) {$\CartaDesc$};
\node (66) at (6,6) {$\CartaDesc$};

\node (0) at (9,0) {$\CartaDesc$};
\node (1) at (10,0) {$\CartaDesc$};
\node (2) at (11,0) {$\CartaDesc$};
\node (3) at (12,0) {$\CartaDesc$};
\node (4) at (13,0) {$\CartaDesc$};
\node (5) at (14,0) {$\CartaDesc$};
\node (6) at (15,0) {$\CartaDesc$};

\node (01) at (9,1) {$\CartaDesc$};
\node (11) at (10,1) {$\CartaDesc$};
\node (21) at (11,1) {$\CartaDesc$};
\node (31) at (12,1) {$\CartaDesc$};
\node (41) at (13,1) {$\CartaDesc$};
\node (51) at (14,1) {$\CartaDesc$};
\node (61) at (15,1) {$\CartaDesc$};

\node (02) at (9,2) {$\CartaDesc$};
\node (12) at (10,2) {$\CartaDesc$};
\node (22) at (11,2) {$\CartaDesc$};
\node (32) at (12,2) {$\CartaDesc$};
\node (42) at (13,2) {$\CartaDesc$};
\node (52) at (14,2) {$\CartaDesc$};
\node (62) at (15,2) {$\CartaDesc$};

\node (03) at (9,3) {$\CartaDesc$};
\node (13) at (10,3) {$\CartaDesc$};
\node (23) at (11,3) {$\CartaDesc$};
\node (33) at (12,3) {$\CartaDesc$};
\node (43) at (13,3) {$\CartaDesc$};
\node (53) at (14,3) {$\CartaDesc$};
\node (63) at (15,3) {$\CartaDesc$};

\node (04) at (9,4) {$\CartaDesc$};
\node (14) at (10,4) {$\CartaDesc$};
\node (24) at (11,4) {$\CartaDesc$};
\node (34) at (12,4) {$\CartaDesc$};
\node (44) at (13,4) {$\CartaDesc$};
\node (54) at (14,4) {$\CartaDesc$};
\node (64) at (15,4) {$\CartaDesc$};

\node (05) at (9,5) {$\CartaCath$};
\node (15) at (10,5) {$\CartaCath$};
\node (25) at (11,5) {$\CartaCath$};
\node (35) at (12,5) {$\CartaCath$};
\node (45) at (13,5) {$\CartaCath$};
\node (55) at (14,5) {$\CartaDesc$};
\node (65) at (15,5) {$\CartaDesc$};

\node (06) at (9,6) {$\CartaDesc$};
\node[label=above:y] (16) at (10,6) {$\CartaAlicia$};
\node (26) at (11,6) {$\CartaDesc$};
\node (36) at (12,6) {$\CartaDesc$};
\node (46) at (13,6) {$\CartaDesc$};
\node (56) at (14,6) {$\CartaDesc$};
\node (66) at (15,6) {$\CartaDesc$};

\draw[dotted] (0,5) -- (6,5);
\draw (6,0) -- (6,6);
\draw (4.333,0) -- (6,5);
\draw (5,6) -- (6,5);
\draw (1,0) -- (6,5);
\draw (6,5) -- (3.5,0);
\draw (6,5) -- (5,0);

\draw (6,5) -- (4.75,0);

\draw (9,6) -- (15,6);
\draw[dotted] (10,6) -- (9,5);
\draw[dotted] (10,6) -- (15,1);
\draw[dotted] (10,6) -- (15,3.5);
\draw[dotted] (10,6) -- (10,0);
\draw[dotted] (10,6) -- (15,4.333);
\draw (10,6) -- (15,4.75);
\draw (10,6) -- (15,5);

\node (65) at (6,5) {$\CartaAlicia$};

\end{tikzpicture}
\end{center}
\caption{Lines that Cath may discard from Alice's announcement. It is important to note that most of the lines are truncated, as the natural representation of $\mathbbm{F}^2_7$ is as a torus and lines are harder to visualize in two dimensions.}
\label{XYlines}
\end{figure}

First we will take a look at $x$. Cath knows that, in order for Alice to hold $x$, one of the lines that passes through $x$ must be Alice's hand. We draw these lines on the plane.

Cath knows that not all the lines that pass through $x$ can be Alice's hand, because if a line contains a card that belongs to Cath, it clearly cannot be held in its entirety by Alice . In this case, only one line fits that description, so Cath takes it out of consideration. We denote this by drawing the line dotted. Every point in the plane has $8$ lines that cross it; therefore, the point $x$ still belongs to $7$ hands that could possibly belong to Alice.

However, this is not the case for all cards that Cath does not hold. Let us now turn our attention to $y$. While $x$ was colinear with Cath's hand, all the lines that contain $y$ and one of Cath's cards are different. In this case Cath can discard more lines than she could when considering $x$. Only $3$ possible lines remain, compared to the $7$ lines that pass through $x$ and avoid Cath's hand. Therefore, it seems to Cath that the point $x$ would be more likely to belong to Alice's hand than the point $y$ as there are more possible hands that contain it. Before the announcement, both cards had the same probability to be in Alice's hand but after the announcement, $x$ seems far more likely.

Note that the total number on lines in the announcement is $56$. We also know that $36$ of these lines contain a card that Cath holds. This is because there are $8$ lines touching each point, but the $5$ points all share one line. Therefore Alice's hand is one of the $20$ lines that avoid Cath's hand. Of those $20$ only three contain $y$ compared to the $7$ that contain $x$. Thus, it seems to Cath that there is a $\nicefrac{7}{20}=0.35$ probability that Alice holds $x$ compared to $\nicefrac{3}{20}=0.15$ that she holds $y$. Thus, according to the information that Cath has, it is more than twice as likely that Alice holds $x$ as it is that she holds $y$.

In this case, we know neither of the cards actually belongs to Alice, but we want to be able to quantify this information and control it, especially in higher dimensions where it is not as simple to visualize. Our goal is to show that, by choosing different parameters appropriately, we can make the different probabilities be arbitrarily close to each other. But first we need some preliminaries to make this precise.

\section{Strategies and Probabilistic Security}

In this section we will set up the basic concepts needed to formalize the Russian cards problem and different notions of security that one may require from its possible solutions. We will assume that Alice holds $a$ cards, Bob $b$ and Cath $c$, and $\Omega$ is the set of cards with $|\Omega|=a+b+c$. A \textbf{deal (of size $(a,b,c)$)} is a partition $(A,B,C)$ of $\Omega$ such that $|A|=a$, $|B|=b$ and $|C|=c$; each of $A,B,C$ represent the hand of Alice, Bob and Cath, respectively.

\subsection{Equitable strategies}

In most solutions to the Russian cards problem, Alice makes an announcement, after which Bob knows the entire deal and thus can make a second (trivial) announcement where he tells Alice which cards Cath holds. Thus we need only model Alice's first announcement, and we follow \cite{swanson2014combinatorial} in referring to the way that Alice is to choose her announcement as a {\em strategy.}

Suppose that Alice holds $a$ cards, Bob holds $b$ and Cath holds $c$. Given a set $X$ and a natural number $n$, we denote by $X\choose n$ the set of $n$-element subsets of $X$, and we will refer to such sets as {\em $n$-sets.} We denote the cardinality of $X$ by $|X|$. A possible hand for Alice is then an element of $\Omega \choose a$. In Alice's first announcement she gives a set of possible hands that she may hold, and thus we may consider an announcement simply as a set $\mathcal A\subset {\Omega \choose a}$.

However, there are many possible announcements that may inform Bob of Alice's hand, and it may be convenient for Alice to randomize from all such possible announcements. Thus a strategy for Alice consists on a probability distribution among the possible announcements that she may choose from.

\begin{define}
A \textbf{strategy (on $\Omega \choose a$)} is a function $\mathfrak S$ that assigns to each hand $A\in {\Omega \choose a}$ a probability distribution over $2^{\Omega\choose a}$. We denote the probability of an announcement $\mathcal{A}$ given the hand $A$ as $P_{\mathfrak{S}}(\mathcal{A}|A)$.

Given a strategy $\mathfrak{S}$ and a hand $A$, we will say that $\mathcal A$ is a {\em possible announcement} if $P_\mathfrak{S}(\mathcal{A}|A)>0$. The set of possible announcements will be denoted by $\mathfrak{S}_A$.
\end{define}

When it is clear from context, we will drop the subindex $\mathfrak S$ and write simply $P(\mathcal{A}|A)$ to simplify notation. It will also be convenient for computations if the number of possible announcements is independent of Alice's hand. If we could guarantee that there are $\ParEqui$ possible announcements for each hand, we could always assign a probability of $\nicefrac  1\ParEqui$ to each individual announcement. If a strategy has this property, we will say it is equitable \cite{swanson2014combinatorial}.

\begin{define}
A strategy $\mathfrak S$ is \textbf{equitable} if there exists a positive integer $\ParEqui$ such that, for every $a$-set $A$, $|\mathfrak{S}_A|=\ParEqui$ and the probability of choosing a particular announcement $\mathcal{A} \in \mathfrak{S}_A$ is $P(\mathcal{A}|A)=1/\ParEqui$.
\end{define}

One advantage of equitable strategies is that we need less information to specify them than more general strategies. In particular, we may model equitable strategies merely as a function
\[\mathfrak S\colon\textstyle{\Omega\choose a}\to 2^{2^{\Omega\choose a}},\]
where $\mathfrak S_A$ is the set of announcements with positive probability (and thus with probability $\nicefrac 1\ParEqui$). Since the geometric strategy, which will be our main focus, is equitable, we will adopt this presentation.

The first condition that a two-step solution to the Russian cards problem should satisfy is that Bob should be informed of Alice's hand after an announcement. Let us make this precise. First, we introduce an abuse of notation that we will use throughout the text.

If $\mathcal X\subset 2^\Omega$ and $Y\subset \Omega$, define
\[\mathcal X\setminus Y=\{X\in \mathcal X : X\cap Y=\varnothing\}.\]
Thus, $\mathcal X\setminus Y$ is the set of elements of $\mathcal X$ avoiding $Y$.

\begin{define}
Fix integers $a,b,c$ and a deck $\Omega$ with $|\Omega|=a+b+c$. A strategy $\mathfrak S$ on $\Omega\choose a$ is \textbf{informative for $(a,b,c)$} if, for every $A\in {\Omega\choose a}$ and every $B\in {{\Omega\setminus A}\choose b}$, $\mathfrak S_A\setminus B=\{A\}$.
\end{define}
Thus after an informative announcement, Bob knows exactly which hand $A$ Alice is holding. But an informative strategy may also give Cath information, yet we also require for Alice's strategy to be secure.

\subsection{Probabilistic Security}

Before Alice makes an announcement, Cath knows that Alice can possibly hold any hand that doesn't contain one of Cath's cards. Hence, there are ${a+b}\choose{a}$ possible hands for Alice. However, after an announcement, Cath can discard any hand that isn't found in the announcement. After doing so, it is possible that Cath acquires new information about the cards she does not hold. In particular, she may know that there is a high probability that Alice holds a given card. If Alice and Bob want to communicate securely, it would be desirable to avoid giving Cath such information.

There are three different notions of probabilistic security for strategies: weak, perfect, and our notion of $\varepsilon$-strong security, which lies between the other two. Unconditional security is equivalent to weak security. If we wanted to avoid Cath learning any probabilistic information after an announcement, we would need to ensure that no card seems more likely after the announcement than it did before. For this, the number of hands in the announcement (after Cath eliminates the ones which have a card that she holds) that contain a given card must be equal for every card that Cath does not hold. In this case, the probability of Alice having a set card should stay the same after Alice's announcement. As a matter of fact, we know the value of this probability; we must only count the hands that could contain that card given Cath's hand and divide it by the number of remaining hands in the announcement:
\[P(x \in A|C)=\frac{{a+b-1 \choose a-1}}{{a+b \choose a}}=\frac{a}{a+b}.\]
If this number stays constant after Alice's announcement, we will say that Alice's strategy is perfectly secure.

\begin{define}
A strategy $\mathfrak{S}$ on $\Omega\choose a$ is \textbf{perfectly secure for $(a,b,c)$} if for every $C\in{\Omega\choose c}$, every card $x\in \Omega\setminus C$, and every announcement $\mathcal A$ with $P(\mathcal A|C)\not=0$, we have that
\[\frac{P(x \in A|C, \mathcal{A})}{P(x \in A|C)} = 1.\]
\end{define}

This notion is equivalent to $1$-perfect security in \cite{swanson2014combinatorial} and represents Cath's inability to gleam information about the position of individual cards. Compare this to {\em weak security,} where we only require that Cath is not certain about the position of any card she does not hold.

\begin{define}
A strategy $\mathfrak{S}$ on $\Omega\choose a$ is \textbf{weakly secure for $(a,b,c)$} if for every $C\in{\Omega\choose c}$, every card $x\in \Omega\setminus C$, and every announcement $\mathcal A$ with $P(\mathcal A|C)\not=0$, we have that
\[0<\frac{P(x \in A|C, \mathcal{A})}{P(x \in A|C)} < 1.\]
\end{define}

In \cite{swanson2014additional,swanson2014combinatorial}, the authors present examples of perfectly secure strategies when Cath has at most $3$ cards. Due to the rigidity needed to ensure this level of security, it is not clear whether perfectly secure strategies can be constructed when Cath holds more cards. Instead, we will define an intermediate level of security, where the constraint is relaxed so we can have more flexibilty and can work in cases where Cath's hand is larger. In fact, this notion will permit us to find secure protocols for any possible hand size that Cath may hold.

\begin{define}
Let $\varepsilon>0$. A strategy $\mathfrak{S}$ on $\Omega\choose a$ is \textbf{$\varepsilon$-strongly secure for $(a,b,c)$} if for every $C\in{\Omega\choose c}$, every card $x\in \Omega\setminus C$, and every announcement $\mathcal A$ with $P(\mathcal A|C)\not=0$, we have that
\[\left |\frac{P(x \in A|C, \mathcal{A})}{P(x \in A|C)}-1 \right |<\varepsilon.\]
\end{define}

As mentioned above, equitable strategies are useful for simplifying computations. In particular, the above probabilities may be computed by counting. The following result can be found in \cite{swanson2014combinatorial}.

\begin{lema}
Let $\mathfrak S$ be an equitable strategy on $\Omega \choose a$ and $(A,B,C)$ be a deal. Suppose that $C\in{\Omega\choose c}$ and $\mathcal A$ is an announcement with $P(\mathcal A| C)>0$ and $A\in \mathcal A$. Then, $P(A|C,\mathcal{A})=\frac{1}{|\mathcal{A}\setminus C|}$.
\end{lema}

In other words, the probability that $A$ is Alice's hand given Cath's hand $C$ and the announcement $\mathcal{A}$ (when $A$ is a valid hand given $C$) is given by the quotient of one over the number of hands in the announcement that avoid $C$.

Thus the probability of Alice having a set hand $A$ according to Cath is $\nicefrac{1}{|\mathcal{A}\setminus C|}$. However, what we want to calculate is the probability that Alice holds a given card $x$. For this, we introduce a new abuse of notation: for $\mathcal X\subseteq 2^{\Omega}$ and $y\in \Omega$, set
\[\mathcal X_y=\{X\in \mathcal X : y\in X\}.\]
Thus for $Z\subset \Omega$, $\mathcal X_y\setminus Z$ denotes the set of elements of $\mathcal X$ which contain $y$ but avoid $Z$. The following can also be found in \cite{swanson2014combinatorial}.

\begin{lema}
Let $\mathfrak S$ be an equitable strategy on $\Omega \choose a$ and $(A,B,C)$ be a deal. If $z\in \Omega\setminus C$, then
\[P(z\in A|C,\mathcal A)=\frac{|\mathcal{A}_z \setminus C|}{|\mathcal{A} \setminus C|}.\]
\end{lema}

\section{Finite geometries}

The geometric strategy is convenient because it allows us to use many familiar results from linear algebra. One key difference when working over finite fields (instead of, say, $\mathbb R$) is that now it becomes relevant to count the number of points in a subspace, the number of subspaces touching a point, etc. These quantities will be useful more than once in this article.

First, we recall a general result about the cardinalities of finite fields. Results in this section are presented without proof; for a more thorough treatment of finite fields and finite geometry, the reader may consult a text such as \cite{dembowski1997,lidl1997finite}.

\begin{teorema}[Existence and Uniqueness of Finite Fields]
If $q$ is a natural number, there exists a finite field with cardinality $q$ if and only if $q$ is of the form $p^n$, with $p$ a prime and $n$ a positive integer. This field is unique up to isomorphism and is called the \textbf{Galois Field of order $q$}. We will denote it by $\mathbbm{F}_{q}$.
\end{teorema}

Our protocol will be based on affine subspaces of a finite vector space. For brevity, we will refer to an affine space of dimension $\alpha$ as an `$\alpha$-plane', while an $\alpha$-plane passing through the origin (i.e., a linear subspace of dimension $\alpha$) will be referred to as an `$\alpha$-space'.

\begin{define}
\label{alphaespacio}
Let $\mathbb V$ be a vector space over a field $\mathbb F$. We say that $W\subset \mathbb V$ is an \textbf{$\alpha$-space} if it is a an $\alpha$-dimensional subspace of $\mathbb V$.

A subset $U\subseteq \mathbb V$ is an \textbf{$\alpha$-plane} if it is of the form $x+W$, where $x\in\mathbb V$ and $W$ is an $\alpha$-space. Two $\alpha$-planes $X,Y$ are \textbf{parallel} if there exists $y\in\mathbb V$ such that $Y=y+X$.
\end{define}

Thus an $\alpha$-plane is similar to an $\alpha$-space, although it does not necessarily pass through the origin. If $\mathbb F$ is finite, then it is not difficult to count the number of points on an $\alpha$-plane.

\begin{lema}
\label{puntosalpha}
If $q$ is a prime power and $\mathbb V$ is a vector space over $\mathbb F_q$, then any $\alpha$-plane in $\mathbb{V}$ has exactly $q^\alpha$ points.
\end{lema}

Meanwhile, the intersection of two distinct $\alpha$-planes is either empty or a $\eta$-plane for some $\eta<\alpha$, which has the following consequence.

\begin{lema}
\label{Inter}
If $q$ is a prime power, $\mathbb V$ is a vector space over $\mathbb F_q$ and $U,W\subseteq \mathbb V$ are distinct $\alpha$-planes, then
\begin{enumerate}
\item \label{InterCap} $|U\cap W|\leq q^{\alpha-1}$ and
\item \label{InterCup} $|U\cup V|\geq 2q^\alpha-q^{\alpha-1}$.
\end{enumerate}
\end{lema}

We may also fix a point $x$ in our vector space and instead ask how many $\alpha$-planes meet $x$. Here the dimension of $\mathbb V$ will be relevant. We will denote the $\delta$-dimensional vector space over $\mathbb F_q$ by $\mathbb F^\delta_q$.

\begin{define}
Fix a prime power $q$. Then, given positive integers $\alpha\leq \delta$, define
\[\kapa q\delta\alpha =\prod_{i=1}^\alpha\left ( \frac{q^\delta-q^{i-1}}{q^\alpha-q^{i-1}}\right ).\]
\end{define}

\begin{lema}
\label{kalpha}
Let $q$ be a prime power and $\delta\geq \alpha>0$. Then,
\begin{enumerate}

\item \label{kaplhaOne} Given $x\in\mathbb F^\delta_q$, the total number of $\alpha$-planes meeting $x$ is equal to $\kapa q\delta\alpha$.

\item \label{kalphaTwo} Given distinct points $x,y \in \mathbb F^\delta_q$, the number of $\alpha$-planes meeting both $x$ and $y$ is given by
\[\kkapa q\delta\alpha=\left ( \frac{q^\alpha-1}{q^\delta-1}\right )\kapa q\delta\alpha.\]

\item \label{kalphaThree} The total number of $\alpha$-planes in $\mathbb F^\delta_q$ is $q^{\delta-\alpha}\kapa q\delta\alpha$.

\end{enumerate}

\end{lema}

\begin{proof}
For the first claim, we may assume without loss that $x=0$. Fix $q$ and let $k^\delta_\alpha$ denote the number of (ordered) sequences $v_1,\hdots,v_\alpha$ of linearly independent vectors in $\mathbb F^\delta_\alpha$. Observe that, for each $\alpha$-space $W$, there are exactly $k^\alpha_\alpha$ ordered bases for $W$, and thus the total number of $\alpha$-spaces is precisely $\nicefrac {k^\delta_\alpha}{k^\alpha_\alpha}$.

Thus it suffices to find an expression for $k^\delta_\alpha$. For $v_1$, we have $q^\delta-1$ options, since we may choose any vector in $\mathbb F^\delta_q$ save for $0$. For $v_2$, we have $q^\delta-q$ options, since we may choose any vector that does not lie in the line generated by $v_1$. More generally, for $v_i$ we have $q^\delta-q^{i-1}$ options, since we may choose any vector that does not lie in the $(i-1)$ space generated by $v_1,\hdots,v_{i-1}$. Thus the total number of options is
\[k^\delta_\alpha=\prod_{i=1}^\alpha (q^\delta-q^{i-1}),\]
and hence the total number of $\alpha$-spaces is
\[\dfrac {k^\delta_\alpha}{k^\alpha_\alpha}=\dfrac{\prod_{i=1}^\alpha (q^\delta-q^{i-1})}{\prod_{i=1}^\alpha (q^\alpha-q^{i-1})}=\kapa q\delta\alpha.\]

The second claim is proven in a similar fashion. As before, we may assume that $x=0$ and fix $y\not=0$. Then we proceed as above, except that we fix $v_1$ to be $y$. This removes the first factor from both the numerator and the denominator, giving us
\[\kkapa q\delta\alpha=\dfrac{\prod_{i=2}^\alpha (q^\delta-q^{i-1})}{\prod_{i=2}^\alpha (q^\alpha-q^{i-1})}=\left (\dfrac{q^\alpha-1}{q^\delta-1}\right) \kapa q\delta\alpha.\]

Finally, for the third claim, note that if we fix an $\alpha$-space $W$, the $\alpha$-planes parallel to $W$, together with $W$, form a partition of $\mathbb F^\delta_q$, hence there are $\nicefrac{q^\delta}{q^\alpha}=q^{\delta-\alpha}$ of them. But there are $\kapa q\delta\alpha$ different $\alpha$-spaces, and thus the total number of $\alpha$-planes is $q^{\delta-\alpha}\kapa q\delta\alpha$.
\end{proof}

\section{The Geometric Strategy}
\label{GeometricP}

We've informally presented the geometric solution to the Russian Cards problem, and will now formalize it to construct the Geometric Strategy. The protocol we will use is essentially presented in \cite{cordon2013geometric}. The basic idea is to construct a finite vector space where every point represents a different card and Alice's hand forms an $\alpha$-plane. Below, we use $f[X]$ to denote the set $\{f(x)\colon x\in X\}$. Each announcement is parametrized by a {\em suitable map.}

\begin{define}
Fix a prime power $q$, natural numbers $0 < \alpha< \delta$ and $A\in {\Omega\choose a}$. We define a {\em suitable map for $A$} to be a function $f\colon \Omega \to \mathbb F^\delta_q$ such that $f[A]$ is an $\alpha$-plane.

Given a suitable map $f$, we define
\[\mathcal A[f]=\{X\subset \Omega : f[X] \text{ is an $\alpha$-plane}\}.\]
\end{define}

The geometric strategy is then defined by letting Alice choose uniformly from all suitable maps $f$ and announcing $\mathcal A[f]$.

\begin{define}[The geometric strategy]
Let $q$ be a prime power and $0<\alpha<\delta$. Given $a$, $b$ and $c$ such that $a=q^\alpha$ and $a+b+c=q^\delta$, we define the \textbf{geometric strategy (with parameters $q,\alpha,\delta$),} denoted $\mathfrak G=\mathfrak G(q,\delta,\alpha)$, to be the strategy such that $\mathfrak G_A$ is the set of all announcements of the form $\mathcal A [f]$, where $f\colon \Omega\to \mathbb F^\delta_q$ is suitable for $A$, and Alice chooses uniformly from $\mathfrak G_A$.
\end{define}

This strategy generalizes that in \cite{cordon2013geometric} where $\alpha=\delta-1$, although that article also considers the case where Alice holds more than one plane. Let us now show that the strategy is equitable.

\begin{lema}
Let $q$ be a prime power, $0<\alpha<\delta$ and $a=q^\alpha$, and let $\mathfrak G=\mathfrak G(q,\delta,\alpha)$. Then, if $A,A' \in {\Omega\choose a}$, $|\mathfrak{G}_A|=|\mathfrak{G}_{A'}|$. 
\end{lema}
\begin{proof}
Let $A, A'\in {\Omega \choose a}$. To show that $|\mathfrak{G}_A|=|\mathfrak{G}_{A'}|$, we will define a bijection $\Sigma\colon \mathfrak{G}_A\to \mathfrak{G}_{A'}$. As a first step, we will build a function $\sigma$ that permutes the elements of $\Omega$. We know that $|A\setminus A'|=|A'\setminus A|$, so we can find a bijection $s\colon A\setminus A' \rightarrow A'\setminus A.$

Using the function $s$ we will define a permutation $\sigma\colon \Omega \rightarrow \Omega$ given by
$$	\sigma(x)= 
	\left\{
		\begin{array}{ll}
			s(x) & \mbox{if } x\in A\setminus A' \\
			s^{-1}(x) & \mbox{if } x\in A'\setminus A \\
			x & \mbox{otherwise.} 
		\end{array}
	\right.$$
Since $s$ is invertible, $\sigma$ is well-defined, and it is easy to check that $\sigma$ is bijective. 
We then define $\Sigma: \mathfrak{G}_A \rightarrow \mathfrak{G}_{A'}$ given by $\Sigma(\mathcal{A})= \{ \sigma [H] | H \in \mathcal{A} \}.$ Then, it is easy to see that $\Sigma$ has an inverse given by $\Sigma^{-1}(\mathcal B)=\{\sigma^{-1}[H]:H\in \mathcal B\},$ and hence $\Sigma$ is a bijection, so that $|\mathfrak G_{A}|=|\mathfrak G_{A'}|$ as claimed.
\end{proof}

We have now proven that for every hand that Alice can have there is the same number $\ParEqui$ of possible announcements. This permits us to set the probability of a particular announcement to be chosen to $1/\ParEqui$, and thus the geometric strategy is equitable. As mentioned above, this will simplify some computations, even without explicitly computing $\ParEqui$.

In the remainder of this section we will prove that the geometric strategy gives an informative and weakly safe solution to the Russian cards problem, provided $c$ satisfies certain bounds.

\begin{lema}
\label{CotaInfo}
Let $q$ be a prime power, $1\leq \alpha <\delta$ and $a,b,c$ positive integers such that $a=q^\alpha$, $a+b+c=q^\delta$ and $c < q^{\alpha}-q^{\alpha-1}$. Then, the geometric strategy with parameters $q,\delta,\alpha$ is informative for $(a,b,c)$.
\end{lema}

\begin{proof}
Let $q,\delta,\alpha,a,b,c$ satisfy the hypotheses of the lemma. Let $A\in{\Omega\choose a}$ and $B\in{{\Omega\setminus A}\choose b}$ and $f\colon \Omega\to \mathbb F^\delta_q$ be such that $f[A]$ is an $\alpha$-space. Clearly $A\in \mathfrak G_A$, so it remains to check that if $A'\in{{\Omega\setminus B} \choose a}$ is such that $\mathcal A[f]\in\mathfrak G_{A'}$, then $A=A'$.

If this were not the case, then $U=f[A]$ would be an $\alpha$-space different from $U'=f[A']$. Since both $U$ and $U'$ are disjoint from $f[B]$, then so is $U\cup U'$. By Lemma \ref{Inter}.\ref{InterCup}, $|U\cup U'|\geq 2q^{\alpha}-q^{\alpha-1}$. But $f$ is a bijection, so it follows that $a+c=|A|+|C|\geq 2q^{\alpha}-q^{\alpha-1}$, and thus $c\geq q^\alpha-q^{\alpha-1}$, contradicting our hypothesis.

We conclude that $U=U'$, so that also $A=A'$ and thus the geometric strategy is informative.\end{proof}

Next we must see that, given a card $x$ not held by Cath, there is a nonzero probability that Alice holds $x$, which means that it is impossible that there is $x\in\Omega\setminus C$ such that all $\alpha$-spaces passing through $x$ meet $C$.

\begin{lema}
\label{cotaAC2}
Let $q$ be a prime power, $0<\alpha\leq \delta$ and $a,b,c$ be such that $a=q^\alpha$ and $a+b+c=q^\delta$. Let $C\in{\Omega \choose c}$, and $f\colon \Omega\to \mathbb F^{\delta}_q$ be a bijection. If $c< \frac{q^\delta-1}{q^\alpha-1}$ and $x\in \Omega\setminus C$, there is $A\in \mathcal A[f]$ such that $x\in A$.
\end{lema}

\begin{proof}
Let $x\in \Omega \setminus C$. Recall that $\kapa q\delta\alpha$ is the number of $\alpha$-planes meeting $f(x)$. If we take $z\in C$, there are $\kkapa q\delta\alpha$ $\alpha$-planes touching both $x$ and $z$ and, since there are $c$ points in $C$, there are at most $c\kkapa q\delta\alpha$ $\alpha$-planes touching $f(x)$ and meeting $f[C]$. Thus in order to guarantee that there is at least one $\alpha$-plane touching $f(x)$ but not $f[C]$, it suffices to have $c \kkapa q\delta\alpha< \kapa q\delta\alpha.$ Solving for $c$ and using Lemma \ref{kalpha}.\ref{kalphaTwo}, this becomes
\[c < \frac{\kapa q\delta\alpha}{\kkapa q\delta\alpha}=\frac{q^\delta-1}{q^\alpha-1}.\]
Thus if $c$ satisfies this constraint, there is an $\alpha$-space $U$ touching $f(x)$ but not meeting $f[C]$ and $A=f^{-1}[U]$ is an element of $\mathcal A[f]$ containing $x$ but disjoint from $C$, as desired.
\end{proof}

Likewise, there should be a nonzero probability that any card not held by Cath is held by Bob. In other words, if $y$ is not held by Cath, there should be an $\alpha$-plane avoiding $y$ {\em and} Cath's hand.

\begin{lema}
\label{AC3}
Let $q$ be a prime power, $0<\alpha\leq \delta$ and $a,b,c$ be such that $a=q^\alpha$ and $a+b+c=q^\delta$. Let $C\in{\Omega\choose c}$ and $f\colon \Omega\to \mathbb F^{\delta}_q$ be a bijection. If $c< q^{\delta-\alpha}$ and $y\in \Omega\setminus C$, there is $A\in \mathcal A[f]$ such that $y\not\in A$ and $A\cap C=\varnothing$.
\end{lema}

\begin{proof}
Let $y\in \Omega\setminus C$ and $z\in C$ be arbitrary and $V$ be an $\alpha$-plane touching both $f(y)$ and $f(z)$. We know that there are $q^{\delta-\alpha}$ $\alpha$-planes parallel or equal to $V$. Thus if $c< q^{\delta-\alpha}$ there is at least one $\alpha$-plane $U$ parallel to $V$ which does not contain any pont from $C$. But by construction, $y$ is not on $U$ either. Thus there is an $\alpha$-plane avoiding both $y$ and $U$, and we may take $A=f^{-1}[U]$.
\end{proof}

\begin{lema}
\label{LemCotaSeg}
Let $q$ be a prime power, $0<\alpha < \delta$ and $a,b,c$ be such that $a=q^\alpha$ and $a+b+c=q^\delta$. Then, the geometric strategy with parameters $q,\alpha,\delta$ is weakly secure for $(a,b,c)$ whenever $c < q^{\delta-\alpha}$.
\end{lema}

\begin{proof}
It is straightforward to check that $q^{\delta-\alpha}<\frac{q^\delta - 1}{q^\alpha-1}$, so this is a direct consequence of Lemmas \ref{cotaAC2} and \ref{AC3}.
\end{proof}

Putting together Lemmas \ref{CotaInfo} and \ref{LemCotaSeg} we obtain the main result of this section.

\begin{teorema}
\label{protgeom}
Let $q$ be a prime power, $0<\alpha < \delta$ and $a,b,c$ be such that $a=q^\alpha$ and $a+b+c=q^\delta$. Then, the geometric strategy with parameters $q,\alpha,\delta$ is weakly safe and informative for $(a,b,c)$ whenever $c<\min(q^{\alpha}-q^{\alpha-1},q^{\delta-\alpha})$.
\end{teorema}

We may use Theorem \ref{protgeom} to find many tuples $(a,b,c)$ for which the geometric strategy is weakly secure. If Alice holds a line in the plane, then we may take $c$ to be almost as large as $a$:

\begin{corolario}\label{CorWeakBigC}
There are infinitely many values of $a$ such that for any $c\leq a-2$ there is $b<a^2$ such that the geometric strategy is informative and weakly safe for $(a,b,c)$.
\end{corolario}

\proof
Take $\alpha=1$, $\delta=2$ and $q$ an arbitrary prime power and apply Theorem \ref{protgeom}.
\endproof

On the other hand, if $c$ is much smaller, then we can give Alice a higher-dimensional plane to ensure that the number of cards is not too large relative to Alice and Cath's hands.

\begin{corolario}\label{CorWeakSmallB}
Given rational $\rho\in(0,1)$, there are infinitely many values of $a$ such that for any $c<a^{\rho}$ there is $b<a^{1+\rho}$ such that the geometric strategy is informative and weakly safe for $(a,b,c)$.
\end{corolario}

\proof
Since $\rho$ is rational, so is $1+\rho$, so we can find $1\leq\alpha<\delta$ such that $1+\rho=\nicefrac \delta\alpha$. Since $\rho<1$, for large enough $q$ we have that $q^{\rho\alpha}<q^{\alpha}-q^{\alpha-1}$. Thus for such a $q$ we may use Theorem \ref{protgeom} to see that, for $a=q^\alpha$, $c<q^{\delta-\alpha}$ and $b=q^\delta-a-c$, the geometric strategy is informative and weakly safe. Moreover, we have that $b<q^{\delta}=q^{(1+\rho)\alpha}=a^{(1+\rho)}$, whereas $c<q^{\delta-\alpha}=q^{\rho\alpha}=\alpha^{\rho}$ was arbitrary, so all desired conditions are met.
\endproof

Observe that in either case, the geometric strategy gives infinitely many solutions for tuples $(a,b,c)$ with $c<a$ and $b< ac$.

\section{Strong safety of the geometric strategy}

Since the geometric strategy is equitable, we may apply the results in the previous section to it in order to find parameters for which this strategy is $\varepsilon$-strongly safe. As we have seen, this strategy is weakly safe if $c<q^\alpha-q^{\alpha-1}$ and $c<q^{\delta-\alpha}$. Our goal will be to find tuples for which it is $\varepsilon$-strongly safe for a given $\varepsilon$.

\subsection{Some auxiliary estimates}

We will need to find bounds on the number of hands that Cath considers possible. We begin by counting the total number of hands in an announcement. The following is a direct consequence of Lemma \ref{kalpha}.

\begin{lema}\label{LemSizeAnn}
The number of $a$-sets in an announcement $\cal A$ of the geometric strategy with parameters $q,\delta,\alpha$ is $q^{\delta-\alpha} \: \kapa q\delta\alpha $.
\end{lema}

Now let us see how many hands Cath can discard from this announcement. Recall that $\mathcal{A}\setminus C$ denotes the set of lines avoiding $C$ and $\mathcal{A}_x \setminus C$ denotes the set of lines avoiding $C$ that also pass through $x$. We may compute the probability that Alice holds $x$ from Cath's perspective as
\[P(x\in A|C,\mathcal{A})=\frac{|\mathcal{A}_x \setminus C|}{|\mathcal{A} \setminus C|}.\]
What we are interested in is bounding the quotient of Cath's perceived probabilites before and after the announcement, that is,
\begin{equation}\label{EqCociente}
\frac{P(x\in A|C,\mathcal{A})}{P(x\in A|C)}=\frac{\nicefrac{|\mathcal{A}_x\setminus C|}{|\mathcal{A}\setminus C|}}{\nicefrac{a}{a+b}}.
\end{equation}
As we will see, by modifying the parameters, this quotient can become arbitrarily close to $1$.

In order to find bounds for \eqref{EqCociente}, it suffices to bound the numerator, since the denominator is constant. Thus we need to estimate $|\mathcal{A}_x \setminus C|$ and $|\mathcal{A}\setminus C|$. Let us begin with the latter.

\begin{lema}\label{LemAmenosC}
If $\mathcal A$ is an announcement of the geometric strategy with parameters $\alpha,\delta,q$ and $C\in{\Omega\choose c}$ is non-empty, then
\begin{equation}
\label{EqAmenosC}
\kapa q\delta\alpha (q^{\delta-\alpha}-c) \leq |\mathcal{A}\setminus C|\leq\kapa q\delta\alpha (q^{\delta-\alpha}-1).
\end{equation}
Both equalities hold whenever $c=1$.
\end{lema}

\proof
By Lemma \ref{LemSizeAnn}, $|\mathcal A|=q^{\delta-\alpha}\: \kapa q\delta\alpha $. Thus we may estimate the number of $\alpha$-planes that meet $C$ and subtract to obtain our bounds.

Suppose that $\mathcal A=\mathcal A[f]$. To bound $|\mathcal{A}\setminus C|$ from below, observe that there are $\kapa q\delta\alpha$ $\alpha$-planes passing through each point in $f[C]$ and there are $c$ such points, so that the number of hands in $\mathcal A$ meeting $C$, which is equal to the number of $\alpha$-planes touching $f[C]$, is at most $c\kapa q\delta\alpha$. It follows that 
\[\kapa q\delta\alpha (q^{\delta-\alpha}-c) \leq |\mathcal{A}\setminus C|.\]
Observe that when $c>1$ we are subtracting one $\alpha$-plane at least twice so the inequality is strict, but when $c=1$ then equality holds.

Now let us show the right-hand inequality.  Since $C\not=\varnothing$, we can pick $z\in C$ and observe that there are $\kapa q\delta\alpha$ $\alpha$-planes meeting $f(z)$, and thus {\em at least} $\kapa q\delta\alpha$ meeting $f[C]$. It follows that
\[|\mathcal{A}\setminus C|\leq \kapa q\delta\alpha q^{\delta-\alpha}-\kapa q\delta \alpha=\kapa q\delta\alpha (q^{\delta-\alpha}-1),\]
and if $C=\{z\}$ this computation is exact.
\endproof

\begin{lema}\label{LemAxC} If $\mathcal A$ is an announcement of the geometric strategy with parameters $\alpha,\delta,q$ and $C\in {\Omega\choose c}$ is non-empty, then
\begin{equation}\label{EqAxC}
\kapa q\delta \alpha-c\kkapa q\delta\alpha \leq|\mathcal{A}_x \setminus C|\leq \kapa q\delta\alpha-\kkapa q\delta \alpha.
\end{equation}
Equality holds when $c=1$.
\end{lema}

\proof
Once again suppose that $\mathcal A=\mathcal A[f]$. First let us bound $|\mathcal{A}_x \setminus C|$ from below. To give our estimate, we will take the number of $\alpha$-planes passing through $f(x)$ and subtract the number of $\alpha$-planes passing through $f(x)$ and $f(y)$ for each $y\in C$, {\em without} taking into account that many $\alpha$-planes will be subtracted twice. Evidently this bound will not be tight, but it will be sufficient to establish our main results.

Recall that $\kapa q\delta\alpha$ counts {\em all} of the $\alpha$-planes passing through a given point. But, there are $\kkapa q\delta \alpha$ $\alpha$-planes passing through $f(x)$ and $f(y)$ for each $y\in C$, and thus there are at most $c\kkapa q\delta \alpha $ planes meeting both $f(x)$ and $f[C]$, hence also hands in $\mathcal A$ meeting $x$ and $C$. It follows that
\[\kapa q\delta\alpha- c \kkapa q\delta \alpha  \leq |\mathcal{A}_x \setminus C|.\]

Now let us bound $|\mathcal{A}_x \setminus C|$ from above. This time we use the fact that there is at least one $y_0\in C$, and we can discard all of those $\alpha$-planes that touch $f(y_0)$ as well as $f(x)$, of which there are $\kkapa q\delta \alpha$. It follows that
\[|\mathcal{A}_x \setminus C|\leq \kapa q\delta \alpha-\kkapa q\delta \alpha,\]
and the result follows. 
Once again this bound is exact when $c=1$.
\endproof

\subsection{Bounding probabilities}

The counting lemmas we have given above may be used to bound the probabilities we are interested in. First, we give a more exact bound, and later we will give a simplified version.

\begin{lema}\label{LemCotaFuerte}
If $\mathcal A$ is an announcement of the geometric strategy with parameters $q,\alpha,\delta$ and $C\in{\Omega\choose c}$ 
is non-empty,
\begin{equation}
\label{EqProbInf}
\frac{q^{2\delta}-cq^{\delta+\alpha}-q^\delta+c^2(q^\alpha-1)+c}{q^{2\delta}-q^{\delta}-q^{\delta+\alpha}+q^\alpha}  \leq\frac{P(x\in A|C,\mathcal{A})}{P(x\in A|C)}
\end{equation}
and
\begin{equation}\label{EqProbSup}
\frac{P(x\in A|C,\mathcal{A})}{P(x\in A|C)} \leq\frac{q^{2\delta}-cq^\delta-q^{\delta+\alpha}+cq^\alpha}{q^{2\delta}-cq^{\delta+\alpha}-q^{\delta}+cq^\alpha}.
\end{equation}
\end{lema}

\proof
For the lower bound, we put the lower bound of Lemma \ref{LemAxC} together with the upper bound of Lemma \ref{LemAmenosC} to obtain
\[
\frac{\kapa q\delta\alpha-\kkapa q\delta\alpha c}{\kapa q\delta\alpha (q^{\delta-\alpha}-1)}\leq
\frac{|\mathcal{A}_x \setminus C|}{|\mathcal{A}\setminus C|}
=P(x\in A|C,\mathcal{A}).
\]
Using Lemma \ref{kalpha}.\ref{kalphaTwo} and simplifying we obtain
\[
\frac{q^\delta-cq^\alpha+c-1}{q^{2\delta-\alpha}-q^{\delta-\alpha}-q^\delta+1}\leq P(x\in A|C,\mathcal{A}).
\]
Thus,
\[\frac{\frac{q^\delta-cq^\alpha+c-1}{q^{2\delta-\alpha}-q^{\delta-\alpha}-q^\delta+1}}{\frac{a}{a+b}}\leq\frac{P(x\in A|C,\mathcal{A})}{P(x\in A|C)};\]
using the equalities $a=q^\alpha$ and $a+b+c=q^\delta$ and simplifying once again we obtain
\[
\frac{q^{2\delta}-cq^{\delta+\alpha}-q^\delta+c^2(q^\alpha-1)+c}{q^{2\delta}-q^{\delta}-q^{\delta+\alpha}+q^\alpha}\leq\frac{P(x\in A|C,\mathcal{A})}{P(x\in A|C)}.
\]

Now we turn to bounding the quotient of probabilities from above. As before, we focus on the numerator, since the denominator is fixed, and use Lemma \ref{LemAxC} bound $|\mathcal{A}_x \setminus C|$ from above and Lemma \ref{LemAmenosC} to bound $|\mathcal{A}\setminus C|$ from below.

Thus we obtain the following upper bound:
\[\frac{|\mathcal{A}_x \setminus C|}{|\mathcal{A}\setminus C|} \leq \frac{\kapa q\delta\alpha-\kkapa q\delta\alpha}{\kapa q\delta\alpha (q^{\delta-\alpha}-c)}.\]
Once again we may use Lemma \ref{kalpha}.\ref{kalphaTwo} and some algebra to obtain
\[\frac{P(x\in A|C,\mathcal{A})}{P(x\in A|C)}\leq\frac{q^{2\delta}-cq^\delta-q^{\delta+\alpha}+cq^\alpha}{q^{2\delta}-cq^{\delta+\alpha}-q^{\delta}+cq^\alpha}.\qedhere\]
\endproof

\begin{figure}
\begin{center}
\begin{tabular}{|c|c|c|c|c|c|c|c|}
\hline 
$\bm a$ & $\bm b$ & $\bm c$ & $\bm q$ & $\bm \alpha$ & $\bm \delta$ &  {\bf Lower} & {\bf Upper} \\ 
\hline
\hline 
23 & 504 & 2 & 23 & 1 & 2 &   0.9547 & 1.0456\\ 
\hline 
 43 & 1,803 & 3 & 43 & 1 & 2 & 0.9524 & 1.0488\\ 
\hline
64 & 4,028 & 4 & $2^6$ & 1 & 2 &  0.9524 & 1.0492\\ 
\hline
529 & 11,636 & 2 & 23 & 2 & 3 &  0.9545 & 1.0475\\ 
\hline 
1,849 & 77,655 & 3 & 43 & 2 & 3 &  0.9523 & 1.0499\\ 
\hline 
4,096 & 258,044 & 4 &  $2^6$ & 2 & 3 &  0.9523 & 1.0499\\ 
\hline 
12,167 & 267,672 & 2 &  23 & 3 & 4 & 0.9545 & 1.0476\\ 
\hline 
79,507 & 3,339,291 & 3 & 43 & 3 & 4 &   0.9523 & 1.0499\\ 
\hline 
262,144 & 16,515,068 & 4  & $2^6$ & 3 & 4 & 0.9523 & 1.0499\\ 
\hline 
\end{tabular}
\end{center}
\caption{Some choices of parameters for the geometric strategy along with their respective lower and upper bounds given by Lemma \ref{LemCotaFuerte}. These were found by fixing $c\in \{2,3,4\}$, $\alpha\leq 3$ and $\delta=\alpha+1$ and finding the least $q$ for which  the strategy is $0.05$-strongly safe.}
\end{figure}

\begin{ejem}
Let $q=2^{14}$, $\alpha=1$ and $\delta=3$. This gives rise to the triple
\begin{align*}
a&=16\mathord ,384  \\
b&=4\mathord ,398\mathord ,046\mathord ,494\mathord ,716 \\
c&= 4.
\end{align*}
Although the number of cards is rather large, this triple is remarkable in that it may be considered {\em floating-point perfectly secure;} indeed, by Lemma \ref{LemCotaFuerte} the geometric strategy is $1.118\times 10^{-8}$-strongly safe for this choice of parameters.
\end{ejem}

In fact, when $q$ is large, some simpler bounds will suffice for our purposes.

\begin{corolario}\label{CorSimp}
If $\mathcal A$ is an announcement of the geometric strategy with parameters $\alpha,\delta$ and $C\in{\Omega\choose c}$ is non-empty, then
\begin{equation}\label{CotaSimpInf}
1-\frac{cq^{\alpha}}{q^{\delta}-1}\leq\frac{P(x\in A|C,\mathcal{A})}{P(x\in A|C)}
\end{equation}
and
\begin{equation}\label{CotaSimpSup}
\frac{P(x\in A|C,\mathcal{A})}{P(x\in A|C)}\leq 1+\frac{c(q^\delta+1)+q^{\delta-\alpha}}{q^{2\delta-\alpha}-cq^{\delta}-q^{\delta-\alpha}}.
\end{equation}
\end{corolario}

\proof
Observe that, since $\delta\geq 2$, it follows that $-q^{\delta+\alpha}+q^\alpha<0$; thus we can remove this term from the denominator in \eqref{EqProbInf}, as well as some positive terms from the numerator to obtain
\[\frac{q^{2\delta}-cq^{\delta+\alpha}-q^\delta}{q^{2\delta}-q^{\delta}}\leq\frac{P(x\in A|C,\mathcal{A})}{P(x\in A|C)},\]
which gives us our lower bound \eqref{CotaSimpInf} by simplifying.

We may also obtain a simpler upper bound by removing negative terms from the numerator and positive terms from the denominator, giving us
$$\frac{P(x\in A|C,\mathcal{A})}{P(x\in A|C)}\leq\frac{q^{2\delta}+cq^\alpha}{q^{2\delta}-cq^{\delta+\alpha}-q^{\delta}},$$
which factoring $q^\alpha$ and performing polynomial division becomes our simplified upper bound \eqref{CotaSimpSup}.
\endproof

\subsection{Convergence}

Our simplified bounds from Corollary \ref{CorSimp} will be enough to yield many tuples for which the geometric strategy is $\varepsilon$-strongly safe for arbitrarily small $\varepsilon$. It is based on the following.

\begin{teorema}\label{TheoMain}
Let $\varepsilon>0$, $1\leq \alpha<\delta$ and $\bar c\colon \mathbb N\to\mathbb N$ be such that $\bar c(q)=o(q^{\delta-\alpha})$. Then, if $q$ is a large enough prime power, the geometric strategy with parameters $q,\alpha,\delta$ is $\varepsilon$-strongly safe for any $c<\bar c(q)$.
\end{teorema}

\proof
If $\bar c(q)=o(q^{\delta-\alpha})$ then $1-\frac{\bar c(q)q^{\alpha}}{q^{\delta}-1}$ and $1+\frac{\bar c(q)(q^\delta+1)+q^{\delta-\alpha}}{q^{2\delta-\alpha}-\bar c(q) q^{\delta}-q^{\delta-\alpha}}$ both converge to $1$ as $q\to\infty$. It follows from Corollary \ref{CorSimp} that if $q$ is large and $c<\bar c(q)$,
\[\left |\frac{P(x\in A|C,\mathcal{A})}{P(x\in A|C)}-1 \right |< \varepsilon,\]
which means that the geometric strategy is $\varepsilon$-strongly safe.
\endproof

However, convergence may be quicker or slower depending on how we choose $\bar c$. For example, if we fix $\xi>0$ and take $\bar c(q)=\lfloor q^{\delta-\alpha-\xi} \rfloor$, then this quotient will tend to one, but if $\xi$ is very small we may need a very large number of cards for it to be less than some given $\varepsilon$. More generally, we have the following:

\begin{teorema}
Fix $1\leq \alpha<\delta$, $\xi\in (0,\delta-\alpha)$ and $\bar c\colon\mathbb N\to\mathbb N$ with $\bar c(q)\leq q^{\delta-\alpha-\xi}$. Then, for $q$ a prime power, any announcement $\mathcal A$ of the geometric strategy any card $x$ and any set of $C$ cards with at most $\bar c(q)$ elements,
\[\frac{P(x\in A|C,\mathcal{A})}{P(x\in A|C)}=1+O(\nicefrac 1{q^\xi}).\]
\end{teorema}

\proof
If we take $c=\bar c(q)\leq q^{\delta-\alpha-\xi}$, we have that
\[\frac{cq^{\alpha}}{q^{\delta}-1}\leq \frac{q^{\delta-\xi}}{q^{\delta}-1}=O(\nicefrac 1{q^\xi}),\]
whereas
\[\frac{c(q^\delta+1)+q^{\delta-\alpha}}{q^{2\delta-\alpha}-cq^{\delta}-q^{\delta-\alpha}}\leq \frac{q^{2\delta-\alpha-\xi}+q^{\delta-\alpha-\xi}+q^{\delta-\alpha}}{q^{2\delta-\alpha}-q^{2\delta-\alpha-\xi}-q^{\delta-\alpha}}=O(\nicefrac 1{q^\xi}).\]
The theorem then follows from Corollary \ref{CorSimp}.
\endproof

Here we see a trade-off between keeping Cath's hand relatively large and obtaining a good rate of convergence for our bounds. Observe, however, that the larger $c$ is, the less tight our bounds are, so despite our bounds converging rather slowly there may be smaller examples with a large degree of security.

\section{Choosing good parameters}

In this section we will focus on strategies for finding specific choices of parameters for which the geometric protocol is $\varepsilon$-safe. In particular, we will fix $\varepsilon=0.05$, and use our bounds to find several explicit tuples for which the geometric strategy is $\varepsilon$-safe. It is interesting to compare this to \cite{swanson2014additional}, where many choices of parameters for which the protocol is perfectly safe are exhibited. All of the tuples presented there have $c\leq 3$, and the authors discuss the difficulty of finding perfectly safe strategies for larger $c$. As we shall see this becomes substantially simpler if we weaken our requirements to (e.g.) $0.05$-strong safety. Thus we may argue that passing to a weaker notion of security allows us to make the Russian cards problem substantially easier to solve without compromising security in a practically meaningful way.

\subsection{Cath has one card}
The notion of perfect security for the Russian cards problem was introduced in \cite{swanson2014combinatorial}, where several examples with $c=1$ are provided. Here we will show that, in this setting, the bounds we have found may also be used to establish perfect security in the case of the geometric strategy.
\begin{corolario}\label{perfect}
Given $1\leq \alpha<\delta$ and a prime power $q$, the geometric strategy is perfectly safe for $(a,b,1)$ with $a=q^\alpha$ and $b=q^\delta-a-1$.
\end{corolario}

\proof
We will use our original bounds from Lemma \ref{LemCotaFuerte}. Setting $c=1$ in \eqref{EqProbInf} we see that
$$\frac{P(x\in A|C,\mathcal{A})}{P(x\in A|C)}\leq\frac{q^{2\delta}-q^\delta-q^{\delta+\alpha}+q^\alpha}{q^{2\delta}-q^{\delta+\alpha}-q^{\delta}+q^\alpha}=1.$$
Similarly, we substitute $c=1$ in \eqref{EqProbSup} to obtain
\[1=\frac{q^{2\delta}-q^{\delta+\alpha}-q^\delta+q^\alpha-1+1}{q^{2\delta}-q^{\delta}-q^{\delta+\alpha}+q^\alpha}\leq\frac{P(x\in A|C,\mathcal{A})}{P(x\in A|C)}.\]
Thus we have seen that $\frac{P(x\in A|C,\mathcal{A})}{P(x\in A|C)}= 1,$ that is, ${P(x\in A|C,\mathcal{A})}={P(x\in A|C)}.$
\endproof

Hence our bounds give an alternative proof that the geometric strategy is perfectly secure when $c=1$. Now let us turn our attention to larger $c$.

\subsection{Making Cath's hand large}

Suppose that we wish to obtain good tuples for which $c$ is as large as possible relative to Alice's hand. Cath's hand is bounded by two expressions on $\alpha$, one increasing on $\alpha$ ($c<q^\alpha-q^{\alpha-1}$) and one decreasing ($c=o(q^{\delta-\alpha})$). Thus, the maximum value that $c$ may take is when the two bounds coincide, which occurs approximately when $\delta=2\alpha$ or $\delta=2\alpha+1$.

The latter case is interesting, since we already have that $q^\alpha-q^{\alpha-1}<q^{\delta-\alpha-1}$. Thus in cases that Alice holds relatively few cards, less than the square root of the deck, our informativity bound will already give us $\varepsilon$-strong safety for large values of $q$. To be precise, we have the following:

\begin{corolario}\label{CorAlpha}
Let $\varepsilon,\beta>0$ and $\ParCath\in(0,1)$. Then, there are infinitely many values of $a$ such that for any $c<\ParCath a$ there is  $b<a^{2+\beta}$ so that the geometric strategy is informative and $\varepsilon$-strongly safe for $(a,b,c)$.
\end{corolario}

\proof
Let $\varepsilon,\beta>0$ and $\ParCath\in(0,1)$. Pick $\alpha$ large enough so that $\nicefrac{1}\alpha\leq \beta$ and $Q$ large enough so that $\nicefrac 1Q<\nicefrac {(1-\ParCath)} 2$. Set $\delta=2\alpha+1$ and $\bar c(q)=q^\alpha+q^{\alpha-1}-1$. Then, since $\delta-\alpha=\alpha+1$ we have that $\bar c(q)=o(q^{\delta-\alpha})$, from which it follows from Theorem \ref{TheoMain} that for large $q$, the geometric strategy is informative and $\varepsilon$-strongly safe for $a=q^\alpha$, $c\leq \bar c(q)$ and $b=q^{\delta}-a-c$. In particular we may also take $q>Q$, so that
\[\bar c(q)=q^\alpha-q^{\alpha-1}-1=(1-\nicefrac 1q-\nicefrac 1{q^\alpha})a> \ParCath a.\]
Meanwhile, $b<q^\delta=(q^{\alpha})^{\frac{2\alpha+1}\alpha}<a^{2+\beta},$ so all desired conditions are met.
\endproof

Compare the above result to Corollary \ref{CorWeakBigC}. As before we can have $c=O(a)$, but this time instead of having $b<ac$ we must take $b=O(ac^{1+\beta})$. Thus the price of obtaining $\varepsilon$-strong security is to make Bob's hand a bit larger than we would need for weak security. In Figure \ref{FigAlpha}, we fix values of $\ParCath$ and $\beta$ and use the strategy of Corollary \ref{CorAlpha} and its proof to find tuples for which the protocol is $0.5$-strongly secure.

\begin{figure}
\begin{center}
\begin{tabular}{|c|c|c|c|c|c|c|c|}
\hline 
$\bm a$ & $\bm b$ & $\bm c$ & $\bm q$ & $\bm \alpha$ & $\bm \ParCath$ &{\bf Lower} & {\bf Upper} \\ 
\hline
\hline
 8 & 500 & 4 & $2^3$& 1 & 0.5 &  0.9527 & 1.0438 \\ 
\hline
9&715&5&$3^3$&1&0.5&0.9503&1.0468
\\
\hline
11 & 1,316 & 4  & 11 & 1 & 0.3 & 0.9750 & 1.0233\\ 
\hline 
13 & 2,180 & 4  & 13 & 1 & 0.3 &  0.9821 & 1.0167 \\ 
\hline
17&4,881&15&17&1&0.8 &0.9515&1.0480\\
\hline
19 & 6,836 & 4 & 19 & 1 & 0.2 &  0.9916& 1.0079 \\ 
\hline
49&16,743&15&7&2&0.3&0.9591&1.0418
\\
\hline
125&77,974&26&5&3&0.2&0.9599&1.0414\\
\hline
\end{tabular}
\end{center}
\caption{Tuples for which the geometric strategy is $0.5$-strongly secure with $q$ a prime power, $a=q^\alpha$, $c \geq \ParCath \alpha$, and $b=q^{2\alpha+1}-a-c$.}\label{FigAlpha} 
\end{figure}

\subsection{Making Bob's hand small}

We may also give a $\varepsilon$-strongly secure analogue of Corollary \ref{CorWeakSmallB}. Suppose that we instead want to have a small number of cards in Bob's hand relative to Alice's. In our above construction Bob's hand grew relatively quickly, so we want a different strategy for selecting parameters. The trade-off will be that Cath's hand may be substantially smaller than Alice's. In general if we want Bob to have less than $a^{1+\beta}$ cards, then Cath must have less than $a^\ParCath$ for some $\ParCath<\beta$.

\begin{corolario}\label{CorRho}
Let $\varepsilon>0$, $\ParCath\in(0,1)$ and $\beta>\ParCath$ be such that $\beta\in\mathbb Q$. Then, there are infinitely many values of $a$ such that for any $c<a^{\ParCath}$ there is $b<a^{1+\beta}$ such that the geometric strategy is informative and $\varepsilon$-strongly safe for $(a,b,c)$.
\end{corolario}

\proof
Similar to the proof of Corollary \ref{CorWeakSmallB}, but taking $c<q^{\ParCath\alpha}$ and using Theorem \ref{TheoMain}.
\endproof

Once again, we may obtain $\varepsilon$-strong security for infinitely many tuples $(a,b,c)$ where $c<a$ and $b=O(ac^{1+\beta})$. So the takeaway in either case is that, making Bob's hand only a bit bigger compared to the rest of the deck and choosing an appropriately large deck, we may obtain $\varepsilon$-strong security instead of merely weak security. In Figure \ref{FigRho} we use this idea to find additional tuples for which the geometric protocol is $0.5$-safe.

\begin{figure}
\begin{center}
\begin{tabular}{|c|c|c|c|c|c|c|c|c|}
\hline 
$\bm a$ & $\bm b$ & $\bm c$ & $\bm q$ & $\bm \alpha$ & $\bm \ParCath$ & $\bm \delta$ & {\bf Lower} & {\bf Upper} \\ 
\hline
\hline
16&1,004&4&4&2&0.5&5&0.9525&1.0470
\\
\hline
25&3,095&5&5&2&0.5&5&0.9677&1.0321
\\
\hline
27&6,524&10&3&3&0.7&8&0.9628&1.0373
\\
\hline
107&11,336&6&107&1&0.4&2&0.9528&1.0490
\\
\hline
529 & 11,636 & 2 & 23 & 2 & 0.15 & 3 & 0.9545 & 1.0475\\
\hline
529 & 279,289 & 23 &23&  2	& 0.5& 4 & 0.9583 & 1.0433
\\
\hline
11,449&1,213,588&6&107& 2 & 0.2& 3&0.9528&1.0495\\
\hline
6,561&4,776,375&33&9&4& 0.4& 7&0.9560&1.0459\\
\hline
\end{tabular}
\end{center}
\caption{Tuples for which the geometric strategy is $0.05$-strongly secure with $q$ a prime power, $\delta\geq \alpha>0$, $\ParCath <\nicefrac \delta\alpha-1$, $a=q^\alpha$, $c=\lfloor a^\ParCath \rfloor$, and $b=q^{\delta}-a-c$.}\label{FigRho} 
\end{figure}

\section{Concluding remarks}

While the Russian cards provides a setting in which we may attain unconditionally safe communication, many known solutions to the Russian cards problem are only weakly safe, meaning that they may provide the eavesdropper with probabilistic information. Although perfectly secure solutions are known, it is somewhat difficult to find tuples for which this level of security may be attained. In this paper we have shown, however, that the amount of probabilistic information obtained by the eavesdropper may be controlled, to the extent that in some cases we can obtain a degree of safety indistinguishable from perfect safety for all practical purposes. Weakening the notion of perfect security has led to an infinite number of new tuples for which we may still obtain a high degree of security, and indeed the bounds we have given may be used to analyze the level of security in any instance of the geometric strategy. Moreover, our techniques had the added bonus of replicating some results of \cite{swanson2014combinatorial}.

There are further directions that may be explored. Although our bounds are tight when $c=1$, the larger $c$ is, the less exact they are, which may keep us from identifying many tuples for which we have a high degree of security. These bounds may be improved with a deeper (and, possibly, messier) combinatorial analysis which takes into account collinear points using the inclusion-exclusion principle. In fact, tailor-made bounds can be used for specific values of $c$ that may be of interest.

Finally, this analysis could be generalized further to include other combinatorial constructions, for example considering a wider class of designs. Such efforts could very well lead to more flexible methods of finding tuples for which there are strategies with very high levels of security.


\end{document}